\documentclass[11pt, oneside, english]{article}

\usepackage[english]{babel}
\usepackage[utf8x]{inputenc}
\usepackage[T1]{fontenc}
\usepackage[letterpaper,top=1in,bottom=1in,left=1in,right=1in,marginparwidth=1.75cm]{geometry}
\usepackage{dsfont}
\usepackage{amsmath, amssymb, amsthm, verbatim,enumerate,bbm}
\usepackage{indentfirst, bbm}
\usepackage{appendix}
\usepackage{enumerate}
\usepackage{verbatim}
\usepackage[colorinlistoftodos]{todonotes}
\usepackage[colorlinks=true, allcolors=blue]{hyperref}
\usepackage[nameinlink,capitalise,noabbrev]{cleveref}
\usepackage{subfig}
\crefname{appsec}{Appendix}{Appendices}
\creflabelformat{enumi}{#2textup{#1}#3}

\allowdisplaybreaks
\makeatletter
\newtheorem{theorem}{Theorem}[section]

\newtheorem*{namedtheorem}{\theoremname}
\newcommand{\theoremname}{testing}

\newtheorem{lemma}[theorem]{Lemma}

\newtheorem{proposition}[theorem]{Proposition}

\newtheorem{corollary}[theorem]{Corollary}

\newtheorem*{question*}{Question}

\theoremstyle{definition}

\newtheorem{remark}[theorem]{Remark}

\newtheorem{open}[theorem]{Open Problem}
\theoremstyle{plain}

\newcommand{\nremk}{[n]\backslash[k]}

\newcommand{\inr}[2]{\langle#1,#2\rangle}

\makeatother

\title{Smoothed analysis of the condition number under low-rank perturbations}
\author{
Rikhav Shah \\
UC Berkeley \\
\tt{rikhav.shah@berkeley.edu}
\and 
Sandeep Silwal \thanks{Research supported by the NSF
Graduate Research Fellowship under Grant No. 1122374. 
} \\
MIT \\
\tt {silwal@mit.edu}
}
\date{}

\begin{document}
\maketitle
\global\long\def\R{\mathbb{R}}

\global\long\def\S{\mathbb{S}}

\global\long\def\Z{\mathbb{Z}}

\global\long\def\C{\mathbb{C}}

\global\long\def\Q{\mathbb{Q}}


\global\long\def\P{\mathbb{P}}

\global\long\def\F{\mathbb{F}}

\global\long\def\U{\mathcal{U}}

\global\long\def\V{\mathcal{V}}

\global\long\def\E{\mathbb{E}}

\global\long\def\Ev{\mathscr{Rk}}

\global\long\def\Dg{\mathscr{D}}

\global\long\def\Ndg{\mathscr{ND}}

\global\long\def\Rv{\mathcal{R}}

\global\long\def\Gv{\mathscr{Null}}

\global\long\def\Hv{\mathscr{Orth}}

\global\long\def\Supp{{\bf Supp}}

\global\long\def\Sv{\mathscr{Spt}}

\global\long\def\ring{\mathfrak{R}}

\global\long\def\Bad{{\boldsymbol{B}}}

\global\long\def\supp{{\bf supp}}

\global\long\def\A{\mathcal{A}}

\global\long\def\L{\mathcal{L}}

\newcommand{\event}{\mathcal E}
\newcommand{\st}{\,|\,}

\newcommand{\p}{\mathbb P}

\newcommand{\eps}{\varepsilon}

\newcommand{\1}{\boldsymbol 1}
\newcommand{\N}{\mathcal{N}}

\begin{abstract}
Let $M$ be an arbitrary $n$ by $n$ matrix of rank $n-k$. We study the condition number of $M$ plus a \emph{low-rank} perturbation $UV^T$ where $U, V$ are $n$ by $k$ random Gaussian matrices. Under some necessary assumptions, it is shown that $M+UV^T$ is unlikely to have a large condition number. The main advantages of this kind of perturbation over the well-studied dense Gaussian perturbation, where every entry is independently perturbed, is the $O(nk)$ cost to store $U,V$ and the $O(nk)$ increase in time complexity for performing the matrix-vector multiplication $(M+UV^T)x$. This improves the $\Omega(n^2)$ space and time complexity increase required by a dense perturbation, which is especially burdensome if $M$ is originally sparse. Our results also extend to the case where $U$ and $V$ have rank larger than $k$ and to symmetric and complex settings. We also give an application to linear systems solving and perform some numerical experiments. Lastly, barriers in applying low-rank noise to other problems studied in the smoothed analysis framework are discussed.
\end{abstract}

\section{Introduction}
The smoothed analysis framework as introduced by Spielman and Teng aims to explain the performance of algorithms on real world inputs through a hybrid of worse-case and average case analysis \cite{smoothed_orig}. In this framework, we are given an arbitrary input that is then perturbed randomly according to some some specified noise model. We apply this framework to study the condition number of a matrix perturbed with low-rank Gaussian noise. The condition number is of interest since it influences the behavior of many algorithms in numerical linear algebra, both in theory and in practice.

To give context to our result, recall that the condition number of a $n \times n$ matrix $M$ with singular values $s_1(M) \ge \cdots \ge s_n(M)$ is defined as the ratio $s_1(M)/s_n(M)$. Generally, a condition number is `well behaved' if  $s_1(M)/s_n(M) = n^{O(1)}$. It can be shown that under very mild and natural conditions, we have $s_1(M) \le n^{O(1)}$. For instance, this readily follows from Proposition \ref{prop:s1} if the entries are not too large compared to the size of $M$ or if the entries have sufficient tail concentration far from the origin. Since our random variables are Gaussians, this easily holds. Therefore, the bulk of the work lies in controlling the \emph{smallest} singular value $s_n(M)$. Extending a result of Edelman \cite{edelman}, Sankar, Spielman and Teng showed the following result in \cite{sankarST}:

\begin{theorem}\label{thm:ST}
There is a constant $C > 0$ such that the following holds. Let $M$ be an arbitrary matrix and let $N_n$ be a random matrix whose entries are iid Gaussian. Let $M_n = M+N_n$. Then for any $t > 0$,
\[\p(s_n(M_n) \le t) \le Cn^{1/2}t.\]
\end{theorem}
Later, Tao and Vu generalized the above result where the entries of $N_n$ are independent copies of a general class of random variables that have mean zero and bounded variance \cite{TaoVuCondition1, TaoVuCondition2}. 

\subsection{Motivation for low-rank noise}
The main drawback of these results is that \emph{every} entry of $M$ must be perturbed by independent noise. This means that if such a perturbation was carried out in practice, we would need to first draw $n^2$ random numbers and store them. This is more problematic if $M$ is sparse to begin with and stored in a data structure utilized for sparse matrices. These observations lead us to ask if we can achieve well-conditioned matrices with \textbf{less randomness and less space}.  Our results demonstrate the answer is yes by replacing the dense Gaussian ensemble $N_n$ with a \textit{low-rank} matrix.

To further motivate our work, we note that in the context of smoothed analysis, Theorem \ref{thm:ST} is used to explain the phenomenon that matrices encountered in practice frequently have `well behaved' condition number. For instance, many matrices can arise out of empirical observations or measurements which can be subject to some inherent noise. 

Similarly, low-rank noise is also natural and arises in many scenarios. Low-rank noise has been studied as a  noise model in least squares \cite{app1}, compressed sensing \cite{app2, app3, app4}, and imaging \cite{app5, app6} to name a few applications. In addition, low-rank noise also arises in many applied sciences model, for instance, see the examples in \cite{app7} and references therein where examples are given for the eigenvalue problem $Mx = x + Ex$, for a low rank matrix $E$, arising in scientific modelling. Furthermore, one of the most frequent properties that matrices in data science posses is having low rank (see \cite{ds1, ds2, ds3, ds4} and references within). Thus for these matrices, the traditional smoothed analysis viewpoint of having a dense perturbation cannot apply due to their low rank requirement. 

Hence, an additional motivation of our work is that studying low-rank noise is a natural step in smoothed analysis which we initiate.


\subsection{Our results}
As stated before, we replace the dense Gaussian perturbation $N_n$ with a low-rank perturbation. Our main result is the following.

\begin{theorem}[Theorem \ref{thm:main_rank_k} simplified]\label{thm:rank_k_update_simp}
    Let $1 \le k \le n/2$ and $M$ be a matrix of rank $n-k$.
    Let $U,V$ be $n\times k$ matrices with i.i.d.\ $\mathcal{N}(0,1)$ entries. Then
    \[
    \p\left(
        s_n(M+UV^T)\le \frac{\eps}{{n\,k}}\, {s_{n-k}(M)}
    \,\right)\le C\,\sqrt{\eps}+\exp(-c\,n)
    \]
    so long as $s_{n-k}(M)<n$
    for absolute constants $C,c > 0$.
\end{theorem}

Theorem \ref{thm:rank_k_update_simp}
roughly states that if we add a rank $k$ random perturbation to a rank $n-k$ matrix, then the smallest $k$ singular values of the matrix improve. The advantage of our approach is that the matrices $U,V$ can be stored separately from $M$ using $O(nk)$ space. This is especially useful in the case that $M$ is sparse to begin with and is stored using a data structure optimized for sparse matrices. Furthermore, a matrix vector product operation $(M+UV^T)x$ can be computed in $\text{Time}(Mx) + O(nk)$ time where $\text{Time}(Mx)$ is the time required to compute $Mx$. For instance, when $k = O(1)$, the extra increase in space and time complexity is only $O(n)$. This is a significant improvement in both the space required to store a dense Gaussian random matrix $G$ and computing $(M+G)x$ which are both $\Omega(n^2)$.  We prove Theorem \ref{thm:rank_k_update_simp} in Section \ref{sec:proof} and discuss the dependence on the terms $s_{n-k}(M)$ and $\sqrt{\eps}$ which we show is unavoidable (see remarks \ref{rmk:unavoid1}, \ref{rmk:unavoid2}).

Theorem \ref{thm:rank_k_update_simp} can be generalized in a variety of ways. First, our result carries over to the case where we pick the columns of $U,V$ to be from a rotationally invariant distribution, such as uniform vectors on the unit sphere. We show that our result also carries over to the case where $M$ is symmetric and we pick $U = V$ to preserve symmetry. 

It is natural to ask if a broader family of random variables can be used in Theorem \ref{thm:rank_k_update_simp}. In Section \ref{sec:subgauss} we show that our result \emph{cannot hold} if we pick the entries of $U,V$ to be from the Rademacher distribution. This is in contrast to the dense perturbation case where Gaussian random variables can be replaced with a wide variety of other distributions such as sub-Gaussian random variables (which include Rademachers). On the other hand, we can get an analogous statement to Theorem \ref{thm:main_rank_k} if we allow for complex Gaussian perturbations.

\begin{theorem}[Theorem \ref{thm:main_rank_k_complex} simplified]\label{thm:rank_k_update_complex_simp}
    Let $1 \le k \le n/2$ and $M$ be a matrix of rank $n-k$.
    Let $U,V$ be $n\times k$ complex matrices with real and imaginary parts in each entry drawn independently from $\mathcal{N}(0,1/2)$. Then
    \[
    \p\left(
        s_n(M+UV^T)\le \frac{\eps}{{n\,k}}\, {s_{n-k}(M)}
    \,\right)\le C{\eps}+\exp(-c\,n)
    \]
    so long as $s_{n-k}(M)<n$
    for absolute constants $C,c > 0$.
\end{theorem}

 A further natural question to consider is if the low-rank noise model can be studied in other problems in smoothed analysis. In Section \ref{sec:challenges}, we highlight the challenges that arise when applying low-rank random perturbations to other well studied problems in smoothed analysis such as the simplex method and $k$-means clustering. We show that current analysis methods that work for dense random perturbations for these problems do not carry over to the low-rank case due to the lack of independence.

Lastly, we note that Theorem \ref{thm:rank_k_update_simp} requires that if the input matrix $M$ has rank $n-k$, then perturbation has rank exactly $k$. This condition can be relaxed in a couple of ways; first, if we add a perturbation of rank less than $k$ then the matrix will be singular so there is nothing to study in this case. On the other hand, adding a rank $k' > k$ perturbation to a rank $n-k$ matrix can be thought of as adding a rank $(k'-k)$ perturbation to a full rank matrix since the original matrix plus a rank $k$ perturbation will be full rank with probability $1$. Then as explained further in Section \ref{sec:beyond}, we can obtain the following result in this case.

\begin{theorem}[Theorem \ref{thm:beyond_k} simplified]\label{thm:beyond_k_simp}
Let $M$ be a $n \times n$ real matrix with $\text{rank}(M) = n$, smallest singular value $s_n$, and $U,V \in \mathbb{R}^{n \times k}$ have independent $\mathcal{N}(0,1)$ coordinates. Then for all $\eps \in (0,1)$,
\[
	\p\left(s_n(M+UV^T)\le\frac\eps{\sqrt n}\right)\le
	C\left(
	\sqrt\eps\left(
	1+\frac{2nk}{s_n}
	\right)
	+
	\frac1{(2nk)^{9/4}}
	+
	\exp(-c nk)
	\right).
\]
\end{theorem}

The second way to circumvent Theorem \ref{thm:rank_k_update_simp} is with the use of Weyl’s perturbation inequality.
To see how it applies, consider the case of $k=1$.  Decompose $M=s_n(M)\ell_nr_n^T+M'$ where $\ell_n$, $r_n$ are the left and right singular vectors associated with $s_n(M)$.  Then we can view $M+uv^T$ as a random perturbation of $M'$ (which has rank $n-1$), plus matrix $s_n(M)\ell_nr_n^T$ whose operator norm is at most $s_n$. We can then apply Theorem \ref{thm:rank_k_update_simp} to $M'$ to bound $s_n(M'+UV^T)$ in terms of $s_{n-1}(M')=s_{n-1}(M)$, and then incur an additional additive $s_n(M)$ error by Weyl's inequality.
Since our ideal use case is when $s_n(M)$ is already negligible, the final bound that we get is comparable to the bound from Theorem  \ref{thm:rank_k_update_simp}. 

Finally in Section \ref{sec:systems}, we discuss an application of low-rank perturbations to solving large sparse linear systems and in Section \ref{sec:numerical}, we present numerical evidence for our low-rank error model.

\begin{remark}
Note that Theorem \ref{thm:ST} and the works of Tao and Vu in \cite{TaoVuCondition1, TaoVuCondition2} both prove a statement of the form $\p(s_n(M+E) \le n^{-A}) \le n^{-B}$ where $E$ is the perturbation and $A,B$ are parameters that depend on the random variables comprising the perturbation. Our statements in Theorem \ref{thm:rank_k_update_simp} is also of similar flavor since it shows that $\p(s_n(M+E) \le s_{n-k}(M) \, n^{-A}) \le n^{-B}$. Since the theorems of Sankar, Spielman, and Teng and Tao and Vu have found other applications in smoothed analysis, such as in the analysis of the simplex method and beyond, we envision that our theorem could also find similar applications. We discuss barriers in applying the low-rank noise model to other smoothed analysis problems in Section \ref{sec:challenges}.
\end{remark}

\subsection{Previous techniques and our approach}\label{sec:previous}
In summary, it is difficult to apply previous techniques in our case since we have \emph{shared randomness} across different rows/columns of the matrix. In more detail, all of the previous techniques used to bound the singular values of a random matrix rely on the controlling the distance between a row to the span of the other rows. To see why this is relevant, imagine a singular matrix. In such a case, it is clear that there must exist a row that lies in the span of the other rows. Therefore, controlling the distance from a row to the span of the other rows gives control over the smallest singular value.

Controlling this geometric quantity boils down to understanding the dot product between a row and the normal vector of the hyperplane spanned by the other rows. Crucially if the rows are independent, we can treat the normal vector of the hyperplane as fixed so this question reduces to the well known Erdos-Littlewood-Offord anti-concentration inequality and its variants which are used in previous works such as \cite{tao2009random, TaoVuCondition1, TaoVuCondition2}. 

To be more precise, lets consider a high level overview of the proof of Sankar, Spielman, and Teng's Theorem \ref{thm:ST}. Fix a vector $x$ and note that from the identity $s_n(M_n) = \|M_n^{-1}\|$, it suffices to give a tail bound on $\|M_n^{-1}x\|$. By applying an orthogonal rotation and using the rotational invariance of the Gaussian, we can say that 
\[\|M_n^{-1}x\| = \|M_n^{-1}e_1\| = \|c_1\|\]
where $e_1$ is the first basis vector and $c_1$ is the first column of $M_n^{-1}$. From the equation $M_n \cdot M_n^{-1} = I$, it follows that $\|c_1\| = 1/|w^Tr_1|$ where $r_1$ is the first row of $M_n$ and $w^T$ is the normal vector of the span of the rows $r_2, \cdots, r_n$. Therefore, the proof reduces to understanding the dot product between a random vector $r_1$, and another independent vector $w$. In the more general case of Tao and Vu \cite{TaoVuCondition1, TaoVuCondition2}, more elaborate dot product estimates using the Erdos-Littlewood-Offord inequality are needed.

In our case, if we add a rank $1$ perturbation to a matrix, randomness is shared across \emph{all} rows. Therefore, we cannot reduce our problem to understanding the dot product between a random vector and another independent vector since fixing a normal hyperplane of a span of a subset of rows automatically gives information about the rows not considered in the span due to the shared randomness. Hence, it is tricky to apply the spectrum of existing techniques in our case.

To overcome these barriers, we use a completely different method to prove Theorem \ref{thm:main_rank_k}. We first reduce our problem to adding noise to a diagonal matrix by using rotational invariance. Then we employ linear algebraic tools (rather than probabilistic tools), to get an `explicit' representation of the inverse of a matrix after adding rank $k$ noise. After arriving at an explicit representation of the inverse, we are able to compute a probabilistic bound on the smallest singular value. Our proof crucially uses the fact that our low-rank perturbations have Gaussian entries whereas the proofs of the dense perturbations carry over in other distributional settings as well. This is not a flaw of our method since it is simply not possible to prove an analogue of our theorems even if the entries of the low-rank perturbations come from sub-Gaussian distributions. We elaborate this point further in Section \ref{sec:subgauss}. 

For Theorem \ref{thm:beyond_k} where we add a rank $k$ random noise to a rank $n$ matrix, we carefully adapt the geometric ideas utilized in previous approaches as explained above. However to get around the shared randomness between rows, we have to perform some careful conditioning which allows us better control the behavior of the random normal vector $w$. 


\subsection{Why would the perturbed matrix even be full rank?}
We briefly address the question of why we even expect low-rank perturbations to improve the condition number. Consider the case where $D$ is a diagonal matrix of rank $n-1$ and we add a random rank $1$ Gaussian perturbation $uv^T$. Recall the matrix determinant lemma which states that 
\[\det(D+uv^T) = \det(D) + v^T \text{adj}(D) u\]
where $\text{adj}(D)$ is the adjugate matrix of $D$. In our case, we can assume that the first $n-1$ entries on the diagonal of $D$ are given by $s_1(D), \cdots, s_{n-1}(D)$ while the last entry is $0$. Then, the adjugate matrix is the all zeros matrix except the bottom rightmost entry which is $s_1(D)\cdots s_{n-1}(D)$. Therefore, 
\[\det(D)+ v^T \text{adj}(D) u = (u_nv_n)(s_1(D)\cdots s_{n-1}(D)) \] 
which is non-zero with probability $1$ since $u_nv_n \ne 0$ with probability $1$. Thus, adding a random rank $1$ perturbation results in $D$ not being singular which motivates the question of studying the smallest singular value after a random rank $1$ (and more generally low-rank) perturbation.

\subsection{Related works}
The smoothed analysis framework has been applied to a variety of problems, most notably in analyzing optimization problems such as $k$-means \cite{smoothedkmeans_orig, smoothedkmeans_best}, the perceptron algorithm \cite{perceptron}, and the simplex method \cite{smoothed_orig, friendlysmooth}. In all of these results, the goal is to show that after an input instance of the problem is suitably perturbed, the algorithm or heuristic runs in polynomial time (the time may depend on the properties of the noise). For a survey of results, see \cite{survey1, survey2, survey3} and references within. The analysis used tends to be very problem specific and also heavily dependent on the type of noise added which for a vast majority of cases are dense Gaussian noise.

\paragraph{Zero preserving noise.} The work that is closest in spirit to our work is the zero preserving noise model studied by Spielman, and Teng. It was shown in \cite{sankarST} that if $M$ is a \emph{symmetric} matrix, then adding an independent Gaussian random variable $x_{ij}$ to each entry $M_{ij}$ such that $i\ne j$, $M_{ij} \ne 0$, and satisfying $x_{ij} = x_{ji}$ along with a Gaussian perturbation along the diagonal results in a `well behaved' condition number. However, the main drawback of this result is that it only holds for symmetric matrices and even in this case, a dense perturbation is required if $M$ is dense to begin with.

\paragraph{Other works.} There are works that use sparse Gaussian perturbations, i.e, their perturbation model is a Gaussian times a Bernoulli random variable with a small parameter. If the Bernoulli random variable has sufficiently small parameter, then with high probability, most of the entries in the perturbation will be zero. The downsides of these methods are that many random variables still need to be drawn and it is not clear if they can show the resulting matrix is well conditioned. For example, Theorem $3.6$ in \cite{reviewer} only shows that the singular values of the resulting matrix are \emph{separated} from each other, not that the matrix is well conditioned. In fact, the study of these types of random matrices where entries are sub-Gaussian random variables multiplied by independent Bernoulli variables is still lacking. For instance, the smallest singular value of such family of random matrices was only recently resolved in the highly technical paper of Basak and Rudelson \cite{BASAK2017426}.

\subsection{Notation}
We use capital letters as $A,M$ to denote matrices and lower case letters such as $x$ for vectors. For a vector $x$, the norm $\|x\|$ is always the Euclidean norm whereas for a matrix $A$, the norm $\|A\|$ always refers to the operator norm (the largest singular value). For a matrix $A$, let $A_S$ denote the sub-matrix of $A$ which includes the $i$th row of $A$ if and only if $i\in S$. The relation $a \lesssim b$ denotes that $a$ is less than or equal to $b$ up to some fixed positive constant and similarly, $a \simeq b$ denotes that $a$ and $b$ are equal up to some fixed positive constant. Unless otherwise indicated, variables $C, c, C_1, C_2, \cdots$ denote positive constants.

\section{Preliminaries}
In this section we enumerate some useful results. First, we recall a classical estimate of the operator norm of a random matrix of Seginer \cite{seginer}. The following proposition essentially shows that the top singular value of a random matrix is well behaved under mild assumptions. Alternatively, one can also bound the top singular value by the frobenius norm if the random variables populating the matrix have sufficient tail concentration.

\begin{proposition}\label{prop:s1}
Let $M$ be a random $n \times n$ matrix with entries $m_{ij}$. Then,
\[\E\|M\| = O\left( \E \max_{1 \le i \le n} \sqrt{\sum_{j=1}^n m_{ij}^2} + \E \max_{1 \le j \le n} \sqrt{\sum_{i=1}^n m_{ij}^2} \right).\]
\end{proposition}

Next we establish tail bounds for the smallest and largest singular values of real and complex Gaussian matrices.

\begin{lemma}[Theorem \ref{thm:ST} reformulated.]
\label{lem:tail_sk}
Let $G \in \mathbb{R}^{k \times k}$ with all entries chosen i.i.d.\ from $\N(0,1)$. Then
    \[\p\left(s_k(G)\le t_1/\sqrt{k}\right)<Ct_1.\]
for some absolute constant $C$.
\end{lemma}

\begin{lemma}[Theorem 1.1 in \cite{tao2009random}]
\label{lem:tail_sk_complex}
Let $G \in \mathbb{C}^{k \times k}$ with all entries chosen with i.i.d.\ real and imaginary parts from $\N(0,1/2)$. Then
    \[\p\left(s_k(G)\le t_1/\sqrt{k}\right)<t_1^2.\]
\end{lemma}

\begin{lemma}[Proposition $2.3$ in \cite{RV}]
\label{lem:tail_s1}
Let $G \in \mathbb{R}^{(n-k)\times k}$ for $k\le n/2$ with all entries chosen i.i.d.\ from $\N(0,1)$. Then 
    \[\p\left(s_1(G)\ge t_2\sqrt{n-k}\right)<C_1e^{-C_2\,t_2^2\,n}.\]
for $t_2$ larger than some absolute constant, and $C_1,C_2$ absolute constants.
\end{lemma}

\begin{lemma}
\label{lem:tail_s1_complex}
Let $G \in \mathbb{C}^{(n-k)\times k}$ for $k\le n/2$ with all entries chosen with i.i.d.\ real and imaginary parts from $\N(0,1/2)$. Then 
    \[\p\left(s_1(G)\ge t_2\sqrt{n-k}\right)<2C_1e^{-C_2\,t_2^2\,n}.\]
for $t_2,C_1,C_2$ as in Lemma \ref{lem:tail_s1}.
\end{lemma}
\begin{proof}
Decompose $G=A+iB$ and bound $s_1(G)\le s_1(A)+s_1(B)$.  Then
\begin{align*}
    \p\left(s_1(G)\ge t_2\sqrt{2(n-k)}\right)
  &\le\p\left(s_1(A)+s_1(B)\ge t_2\sqrt{2(n-k)}\right)
\\&\le\p\left(s_1(A)\ge\frac{t_2\sqrt{2(n-k)}}{2}\right)+\p\left(s_1(B)\ge\frac{t_2\sqrt{2(n-k)}}{2}\right)
\\&\le\p\left(s_1(\sqrt2A)\ge t_2\sqrt{n-k}\right)+\p\left(s_1(\sqrt2B)\ge t_2\sqrt{n-k}\right)
\\&\le 2C_1e^{-C_2\,t_2^2\,n}
\end{align*}
where the last inequality follows by Lemma \ref{lem:tail_s1} since $\sqrt{2}A$ and $\sqrt2B$ have real i.i.d.\ $\mathcal{N}(0,1)$ entries.
\end{proof}

The following lemma bounds the smallest singular value of a block matrix.

\begin{lemma}
\label{lem:op_norm_of_block_inverse}
Let
\[M=\begin{bmatrix}A&B\\C&D\end{bmatrix}\]
be an $n\times n$ matrix.  Then
\[
s_n(M)^{-1}\le \|A^{-1}\|+\|(M/A)^{-1}\|\left(1+\|A^{-1}B\|\right)\left(1+\|CA^{-1}\|\right)
\]
where $(M/A)=D-CA^{-1}B$ is the Schur complement of $A$.
\end{lemma}
\begin{proof}
We first use the Schur formula for the inverse of a block matrix:
\[
M^{-1}
=
\begin{bmatrix}
A^{-1}+A^{-1}B(M/A)^{-1}CA^{-1}&A^{-1}B(M/A)^{-1}\\
(M/A)^{-1}CA^{-1}&(M/A)^{-1}
\end{bmatrix}.
\]
The norm of $M^{-1}$ is upper bounded by the sum of the norms of each of its blocks.
\begin{align*}
s_n(M)^{-1}=\|M^{-1}\|
  &\le\|A^{-1}\|+\|A^{-1}B\|\|(M/A)^{-1}\|\|CA^{-1}\|
\\&  +\|A^{-1}B\|\|(M/A)^{-1}\|
\\&  +\|(M/A)^{-1}\|\|CA^{-1}\|
\\&  +\|(M/A)^{-1}\|
\\&  = \|A^{-1}\|+\|(M/A)^{-1}\|\left(1+\|A^{-1}B\|\right)\left(1+\|CA^{-1}\|\right). \qedhere
\end{align*}
\end{proof}

Lastly, we recall that Gaussians are sufficiently anti-concentrated. 
\begin{proposition}\label{prop:gaussian_ball}
Let $x \sim \N(0,1)$. Then,
$\p(|x| \le \eps) = \Theta(\eps)$ for $\eps$ sufficiently small.
\end{proposition}

\section{Proof of main theorems}\label{sec:proof}
The goal of this section is to prove the following theorem and its complex and symmetric analogs.
\begin{theorem}
\label{thm:main_rank_k}
    Let $M$ be an arbitrary matrix of rank $n-k\ge n/2$.  Let $U,V$ be $n\times k$ matrices with i.i.d.\ $\mathcal{N}(0,1)$ entries.  Then
    \begin{equation}\label{eq:main}
        \p\left(s_n(M+UV^T)\le \frac{t_1^2}{k}\min\left(\frac1{2},\frac{s_{n-k}(M)}{4\,t_2^2\,(n-k)}\right)\right)\le C_1\,t_1+\,C_2\exp({-C_3\,t_2^2\,n})
    \end{equation}
    for $t_1 \le C_4$ and $t_2 \ge C_5$ for some absolute constants $C_i, 1 \le i \le 5$.
\end{theorem}
Our strategy to prove Theorem \ref{thm:main_rank_k} will reduce general $M$ to the case of $M$ nonnegative and diagonal, then express $s_n(M+UV^T)$ in terms of the singular values of $M$ and certain sub-matrices of $U$ and $V$, and finally apply tail bounds to said singular values.  We start by proving a lemma that allows us to reduce to the case of $M$ nonnegative and diagonal. As stated in Section \ref{sec:previous}, this is a compltely different proof strategy than the one used in previous works.
\begin{lemma}
\label{lem:diagonal_suffices}
Let $D=\textnormal{diag}(s_n(M),\cdots,s_1(M))$.  Let $U,V$ be as in Theorem \ref{thm:main_rank_k}.  Then the distributions of $s_n(M+UV^T)$ and of $s_n(D+UV^T)$ are identical.
\end{lemma}
\begin{proof}
Let $LDR^T=M$ be the singular value decomposition of $M$.  Then \[M+UV^T=LDR^T+UV^T=L(D+L^TUV^TR)R^T.\]  Left- and right- multiplication by unitary matrices preserves singular values so $$s_n(M+UV^T)= s_n(D+L^TUV^TR).$$ Finally, $U$ and $V$ are rotationally invariant, so $L^TU$ and $R^TV$ are distributed just as $U$ and $V$ are.
\end{proof}
Now we proceed to the main proof.

\begin{proof}[Proof of Theorem \ref{thm:main_rank_k}]
For any matrix $T$, recall that $T_S$ denotes the sub-matrix of $T$ which includes the $i$th row of $T$ if and only if $i\in S$.  
Lemma \ref{lem:diagonal_suffices} shows that we may assume $M$ is nonnegative and diagonal without loss of generality.  We may write $M$ and $M+UV^T$ in block form as
\[M=\begin{bmatrix}0&0\\0&M'\end{bmatrix}\quad\text{and}\quad
M+UV^T=\begin{bmatrix}
U_{[k]}V_{[k]}^T & U_{[k]}V_{\nremk}^T\\\\
U_{\nremk}V_{[k]}^T & M'+U_{\nremk}V_{\nremk}^T
\end{bmatrix}
\]
where $M'$ has no zeros on the diagonal.  We can use Lemma \ref{lem:op_norm_of_block_inverse} to upper bound $s_n(M+UV^T)$.  The factor corresponding to the Schur complement is
\[\left\|\left(M'-U_{\nremk}\left(I-
V_{[k]}^T(U_{[k]}V_{[k]}^T)^{-1}U_{[k]}
\right)V_{\nremk}^T\right)^{-1} \right \|=\|M'^{-1}\|=s_{n-k}(M)^{-1}\]
since $I-
V_{[k]}^T(U_{[k]}V_{[k]}^T)^{-1}U_{[k]}
 = 0$. This is one of the key insights of our proof. Then the resulting bound is
\begin{align*}
    s_n(M+UV^T)^{-1}
    &\le\frac{1}{s_k(U_{[k]})s_k(V_{[k]}^T)}
    +\frac1{s_{n-k}(M)}\left(1+\|(U_{[k]}V_{[k]}^T)^{-1}U_{[k]}V_{\nremk}^T\|\right)\left(1+\|U_{\nremk}V_{[k]}^T(U_{[k]}V_{[k]}^T)^{-1}\|\right)
  \\&=\frac{1}{s_k(U_{[k]})s_k(V_{[k]}^T)}
  +\frac1{s_{n-k}(M)}\left(1+\|V_{[k]}^{-1}V_{\nremk}^T\|\right)\left(1+\|U_{[k]}^{-1}U_{\nremk}\|\right)
  \\&\le\frac{1}{s_k(U_{[k]})s_k(V_{[k]}^T)}
  +\frac1{s_{n-k}(M)}\left(1+\|V_{[k]}^{-1}\|\|V_{\nremk}^T\|\right)\left(1+\|U_{[k]}^{-1}\|\|U_{\nremk}\|\right)
  \\&=\frac{1}{s_k(U_{[k]})s_k(V_{[k]}^T)}
  +\frac1{s_{n-k}(M)}\left(1+\frac{s_1(V_{\nremk})}{s_k(V_{[k]})}\right)\left(1+\frac{s_1(U_{\nremk})}{s_k(U_{[k]})}\right).
\end{align*}
Denote events
\begin{align*}
  \event_1 &=\left( s_1(U_{\nremk})\le t_2\sqrt{n-k}\quad\text{and}\quad s_1(V_{\nremk})\le t_2\sqrt{n-k}\right),   \\
  \event_2 &=\left( s_k(U_{k})\ge t_1/\sqrt{k}\quad\text{and}\quad s_k(V_{k})\ge t_1/\sqrt{k}\right).
\end{align*}
Conditioning on $\event_1$ and $\event_2$, the above bound becomes
\[
s_n(M+UV^T)^{-1}\le\frac1{s_{n-k}(M)}\left(1+\frac{t_2}{t_1}\sqrt{(n-k)\,k}\right)^2+\frac{k}{t_1^2}.
\]
For sufficiently large $n$ (specifically $n\ge6\frac{t_1^2}{k\,t_2^2}$), this becomes
\[s_n(M+UV^T)^{-1}\le\frac{k}{t_1^2}\left(\frac{2\,t_2^2\,(n-k)}{s_{n-k}(M)}+1\right)\le\frac{2k}{t_1^2}\max\left(\frac{2\,t_2^2\,(n-k)}{s_{n-k}(M)},1\right)\]
Taking the reciprocal of both sides yields
\begin{align*}
s_n(M+UV^T)\ge\frac{t_1^2}{2k}\min\left(\frac{s_{n-k}(M)}{2\,t_2^2\,(n-k)},1\right)
\end{align*}

The probability that this bound is violated is upper bounded by the probability that at least one of $\event_1$ or $\event_2$ fail.  We may upper bound this quantity using the union bound:
\begin{align*}
\p(\neg\event_1\vee\neg\event_2)
  &\le\p(\neg\event_1)+\p(\neg\event_2)
\\&\le\p(s_1(U_{\nremk})\ge t_2\sqrt{n-k})
+\p(s_1(V_{\nremk})\ge t_2\sqrt{n-k})
\\&\quad\quad+\p(s_k(U_{[k]})\le t_1/\sqrt{k})
+\p(s_k(V_{[k]})\le t_1/\sqrt{k})
\\&\le2\,C_1t_1+2\,C_2e^{-C_3\,t_2^2\,n}.
\end{align*}
where the last step follows by applying Lemmas \ref{lem:tail_sk} and \ref{lem:tail_s1} twice each.  The factors of $2$ can be subsumed into the constants $C_1$ and $C_2$ giving the final result.
\end{proof}

\begin{theorem}
\label{thm:main_rank_k_complex}
    Let $M,t_1,t_2,C_2,C_3$ be as in theorem \ref{thm:main_rank_k}.  Let $U,V$ be $n\times k$ complex matrices with real and imaginary parts drawn independently from $\mathcal{N}(0,1/2)$.  Then
    \[
        \p\left(s_n(M+UV^T)\le \frac{t_1^2}{k}\min\left(\frac1{2},\frac{s_{n-k}(M)}{4\,t_2^2\,(n-k)}\right)\right)\le 2\,t_1^2+\,C_2\exp({-C_3\,t_2^2\,n}).
    \]
\end{theorem}
Note that the first term on the righthand side is $2t_1^2$ rather than $C_1t_1$ as it was in Theorem \ref{thm:main_rank_k}.
\begin{proof}
The only place the proof differs from the proof of Theorem \ref{thm:main_rank_k} is in the upper bound on $\p(\neg \event_1)$.  Instead of $C_1t_1$, it is simply $t_1^2$ by Lemma \ref{lem:tail_sk_complex}.
\end{proof}

\begin{remark}
Theorems \ref{thm:main_rank_k} and \ref{thm:main_rank_k_complex} hold when instead of sampling $U$ and $V$ independently, simply set $U=V$.
\end{remark}
\begin{proof}
The proof follows almost exactly as before with only a single modification: In Lemma \ref{lem:diagonal_suffices}, the left- and right- singular vectors of symmetric matrices are the same so $L=R$ (so $L^TU=R^TV$).  Optionally, one may note that events $\event_1$ and $\event_2$ are redundant, so one reduces the bound on $\p(\neg\event_1\vee\neg\event_2)$ by a factor of $2$.
\end{proof}

\begin{remark}\label{rmk:unavoid1}
Let us briefly mention why the term $s_{n-k}(M)$ is unavoidable in the statement of Theorem \ref{thm:main_rank_k}. For simplicity, consider $k=1$ and suppose that $M$ is of rank $n-1$ and suppose its smallest nonzero singular value is equal to $\delta$. After adding a rank $1$ term $uv^T$ to $M$, its rank is $n$ with probability $1$. However, if we consider the limit $\delta \rightarrow 0$, then $M+uv^T$ approaches a rank $n-1$ matrix meaning $s_n(M+uv^T) \rightarrow 0$. Hence, any concentration bound such as \eqref{eq:main} must depend on the term $s_{n-1}$.
\end{remark}

\begin{remark}\label{rmk:unavoid2}
The term $t_1$ on the right hand size of \eqref{eq:main} is also unavoidable. For simplicity, consider the case that $M \in \R^{n \times n}$ is a diagonal matrix and all the entries are non zero except the last diagonal entry and consider a rank $1$ symmetric update. Now after a symmetric rank $1$ update, the perturbed matrix is $M+gg^T$ where $g \in \R^n$. The smallest singular value of a symmetric matrix is equivalent to the smallest absolute eigenvalue. From the Raleigh quotient characterization of eigenvalues, we have that \[\p(s_n(M+gg^T) \le t^2) \ge \p( | e_n^T(M+gg^T)e_n | \le t^2).\]
Using the fact that $Me_n = 0$, we have
\[ \p( | e_n^T(M+gg^T)e_n | \le t^2) = \p(z^2 \le t^2)\]
where $z \in \R$ is a standard normal. Finally, $\p(z^2 \le t^2) = \p(|z| \le t) = \Theta(t)$ from Proposition \ref{prop:gaussian_ball} for $t$ sufficiently small.
\end{remark}

\subsection{Sub-Gaussian perturbations}\label{sec:subgauss}

Just as Tao and Vu generalized Theorem \ref{thm:ST} to the case where more general types of random perturbations beyond Gaussian are used, it is of interest to generalize Theorem \ref{thm:main_rank_k} to the case where $U,V$ are from a general family of distributions. A standard choice are mean zero sub-Gaussian distributions since they encompass well known distributions such as the standard Gaussian and $\pm 1$ (Rademacher) random variables. Surprisingly, we show in this case that we cannot state a general statement like Theorem \ref{thm:main_rank_k} unless extra assumptions about the fixed matrix $M$ is made. 

\begin{lemma}\label{lem:subgauss}
Let $u,v \in \R^n$ be vectors with i.i.d.\ entries that are $\pm 1$ (Rademacher) with equal probability. There exists a rank $n-1$ matrix $M$ such that with constant probability, $M+uv^T$ is singular.
\end{lemma}
\begin{proof}
As in the proof of Theorem \ref{thm:main_rank_k}, let $M = LDR$ and we can say $s_n(M+uv^T) = s_n(D+(L^Tu)(v^TR))$. In the case that $u$ and $v$ are Gaussian, rotational invariance implies that $L^Tu$ and $v^TR$ are distributed as $u,v$ respectively. However, this is no longer the case if $u,v$ have entries coming from general sub-Gaussian distributions, such as $\pm 1$. Here, the properties of $L, R$ can have substantial impact on $s_n(M+uv^T)$. 

Suppose that the top left entry of $D$ is $0$. Then, if the first row of $L$ is \emph{sparse}, i.e. has $O(1)$ non zero entries, then it is possible that the first coordinate of $L^Tu$, $(L^Tu)_1,$ is $0$ with constant probability and hence the first row of $D+(L^Tu)(v^TR)$ is all $0$ which implies that $M+uv^T$ is still rank $n-1$ with constant probability. 
\end{proof}

Therefore, a general statement such as Theorem \ref{thm:main_rank_k} in the case of sub-Gaussian distributions is impossible unless extra assumptions are made about the input matrix $M$. However, we note that in the $k=1$ case, if we \emph{assume} every row of $L, R$ are \emph{dense} (say have at least a constant fraction of non-zero entries), then the proof of Theorem \ref{thm:main_rank_k} carries through in the $\pm 1$ case since the two estimates we need (corresponding to the events $\event_1$ and $\event_2$ respectively) are the concentration of the norms of $L^Tu, v^TR$ and each entry being anti-concentrated from $0$ which follows from Erdos-Littlewood-Offord type results. It is not clear when such an assumption is natural.

\subsection{Application to linear systems}\label{sec:systems}
We briefly highlight the importance of the condition number in solving systems of linear equations. If we are interested in solving the system $Ax = b$ where $A \in \mathbb{R}^{n \times n}$ then the condition number of $A$ influences both the stability and runtime of linear systems solving. Much of this material is standard and can be found in \cite{NLA}.
\\

\noindent \textbf{Stability:} If $\tilde{x}$ denotes the result computed by numerical algorithms to the equation $Ax = b$ then it is known that the relative error quantity $\|x-\tilde{x}\|/\|x\|$ satisfies
\[\frac{\|x-\tilde{x}\|}{\|x\|} = O\left(\eps_{\text{machine}} \cdot \frac{s_1(A)}{s_n(A)} \right) \]
where $\eps_{\text{machine}}$ is the machine precision.
\\

\noindent \textbf{Runtime:} One of the most widely used algorithms for solving systems of linear equations, especially large sparse ones that arise often in practice, is the conjugate gradient descent method. If the conjugate gradient descent method is run for $k$ steps, then its convergence scales roughly as
$$ \left(\frac{\sqrt{s_1(A)/s_n(A)}-1}{\sqrt{s_1(A)/s_n(A)}+1} \right)^k \approx \left(1-\frac{2}{\sqrt{s_1(A)/s_n(A)}} \right)^k.$$
Therefore, a larger the condition number means more steps of the conjugate gradient descent method are required. 

The usefulness of our low-rank error model is further supported by the conjugate gradient descent method. As mentioned previously, this iterative method is mainly used for large sparse systems. Thus, a low-rank perturbation that only requires additional linear space and incurs an additive linear increase in cost per iteration is desirable compared to a dense perturbation which makes the original problem infeasible for large systems.

\section{Perturbation beyond rank $k$ }\label{sec:beyond}
\newcommand{\lemmafourpointforub}[1]{
C\left(
\frac{1}{x_2\cdot #1}\left(1+
\frac{{x_1\cdot x_3^{1/2}\cdot\sqrt{nk}}}{s_n(M)}\right)
+\exp\left(-{x_1^2}/4\right)
+ (2/\pi)^{k/2}x_2^k + \exp(-c\,x_3)\right)
}

In this section, we deal with the case that we add a rank $k'$ perturbation to a rank $n-k$ matrix for $k' > k$. In such a case, we simply ignore a rank $k$ portion of the noise and imagine that we are adding a rank $(k'-k)$ perturbation to a general full-rank matrix. This is valid since the original rank $k$ matrix plus the rank $k$ part of the noise will be full rank with probability $1$. Our result then is the following.

\begin{theorem}\label{thm:beyond_k}
	Let $M$ be a $n \times n$ real matrix with $\text{rank}(M) = n>10$ with smallest singular value $s_n$. Let $U,V \in \mathbb{R}^{n \times k}$ have independent $\mathcal{N}(0,1)$ coordinates. Then,
\begin{equation}\label{eq:beyond_k}
	\p(s_n(M+UV^T) \le 1/t) \le
	C\left(
\frac{\sqrt n}{x_2\cdot t}\left(1+
\frac{{x_1\cdot x_3^{1/2}\cdot\sqrt{nk}}}{s_n(M)}\right)
+\exp\left(-{x_1^2}/4\right)
+ (2/\pi)^{k/2}x_2^k + \exp(-c\,x_3)\right)
\end{equation}
	for all $t > 0$, $x_1 \ge 3\sqrt{2 \log (2nk)}, x_3 \ge nk$, and $x_2 \le 1$.
\end{theorem}

We obtain the following corollary under some natural parameter settings.

\begin{corollary}Let $M$ be a $n \times n$ real matrix with $\text{rank}(M) = n$, and $U,V \in \mathbb{R}^{n \times k}$ have independent $\mathcal{N}(0,1)$ coordinates. Then for all $\eps \in (0,1)$,
\begin{equation}
	\p(s_n(M+UV^T)\le\eps/\sqrt n)\le
	C\left(
	\sqrt\eps\left(
	1+\frac{3nk\sqrt{\log(2nk)}}{s_n(M)}
	\right)
	+
	\frac1{(2nk)^{9/4}}
	+
	\exp(-c nk)
	\right).
\end{equation}
\end{corollary}
\begin{proof}
Set $t=\sqrt n/\eps$, $x_1=3\sqrt{\log(2nk)},x_2=\sqrt\eps,x_3=nk$.
\end{proof}
By setting $\eps$ appropriately, we recover the common `theme' of $\p(s_n(M+UV^T) \le n^{-A}) \le n^{-B}$ as in the case of Theorem \ref{thm:ST} and the works of Tao and Vu \cite{TaoVuCondition1, TaoVuCondition2}.

We now proceed to prove Theorem \ref{thm:beyond_k}. Denote $A = M+UV^T$ and note that $\text{rank}(A) = n$ with probability $1$ so $A^{-1}$ exists.  We observe that \eqref{eq:beyond_k} reduces to bounding $\p(\|A^{-1}\| \ge t)$. Our proof is adapted from the proof of Theorem \ref{thm:ST}. However, we need to perform a careful conditioning argument to prove the most important part of the argument which is presented in Lemma \ref{lem:gaussian}.

We begin by handing the case of a single vector.

\begin{lemma}\label{lem:single_vec}
	For any unit vector $y$, we have
\[\p(\|A^{-1}y\|_2 \ge t) \le
\lemmafourpointforub{t}
\]
	for all $t > 0$, $x_1 \ge 2\sqrt{ \log n}, x_3 \ge nk$, and $x_2 \le 1$.
\end{lemma}
\begin{proof}
	Let $Q$ be a rotation that takes $y$ to $e_n$ and denote $QA$ as $A'$. Then,
	\[\|A^{-1}y\|_2 = \|A^{-1}Q^Te_n\|_2 = \|(QA)^{-1}e_n\|_2 = \|c_n\|_2. \]
	where $c_n$ be the last column of $(QA)^{-1}$. From the identity $A'A'^{-1} = I$, we have that $c_n$ is orthogonal to the first $n-1$ rows of $A'$ and has dot product $1$ with the last row of $A'$. Hence, 
	\[ \|c_n\|_2 = \frac{1}{|\inr{w}{r^n}|} \]
	where $r^i$ is the $i$th row of $A'$ and $w$ is the unique unit vector orthogonal to the span of $\{r^1, \cdots, r^{n-1}\}$ (up to to sign).
	This means that 
	 \[ \p(\|A^{-1}y\|_2 \ge t)  = \p(\|c_n\|_2 \ge t) = \p(|\inr{w}{r^n}| \le 1/t). \]
	 The last row $r^n$ of $A'$ is the sum of the last rows of $QM$ and $QUV^T$. Note that the inner product of $w$ with the last row of $QM$ is some fixed parameter; denote it $r$.  Then $QU$ is distributed as $U$ by the rotational invariance of the normal distribution, so the last row of $QUV^T$ is distributed as $Vu^n$ where $u^n  \in \R^{k}$ is a vector of independent Gaussians. 
	 Therefore, $\inr{w}{r^n}$ is distributed as $\inr{w}{Vu^n} + r$, so it suffices to bound the Levy concentration of $\inr{w}{Vu^n}$.  Specifically, we want to show that
	\[\sup_{r \in \mathbb{R}} \p(|\inr{w}{Vu^n} + r| < 1/t)\le
	\lemmafourpointforub{t}
	\]
	where the probability is over the realization of $u^n$ and $V$. This readily follows from Lemma \ref{lem:gaussian}. 
\end{proof}

\begin{proof}[Proof of Theorem \ref{thm:beyond_k}]
Let $s$ be a unit vector chosen uniformly at random from $\mathcal{S}^{n-1}$. By Lemma \ref{lem:single_vec}, we have \[ \p_{A, s}\left(\|A^{-1}s\|_2 \ge t/\sqrt{n}\right)\le
\lemmafourpointforub{t/\sqrt{n}}.
\]
Now with probability $1$, there exits a unique $y$ such that $\|A^{-1}y \|_2 = \|A^{-1}\|$. From Lemma \ref{lem:SVD}, we see that 
$$ \|A^{-1}s \|_2 \ge \| A^{-1}(y^Ts)y\|_2 = |y^Ts| \|A^{-1} \|. $$
Therefore we have
\begin{align*}
\p_{A, s}\left( \|A^{-1}s\|_2 \ge t/\sqrt{n} \right) &\ge \p_{A, s}\left( \|A^{-1} \|  \ge t\text{ and } |s^Ty| \ge 1/\sqrt{n} \right) \\
&=  \p( \|A^{-1} \|  \ge t) \cdot \p_{A,s} \left( |s^Ty| \ge 1/\sqrt{n}\  \big\lvert \  \|A^{-1} \| \ge t \right).
\end{align*}
By the rotational invariance of $s$, we have that 
\[\p( \|A^{-1} \|  \ge t) \cdot \p_{A,s} \left( |s^Ty| \ge 1/\sqrt{n}\  \big\lvert \  \|A^{-1} \| \ge t \right) = \p( \|A^{-1} \|  \ge t) \cdot \p \left( |s^Te_1|  \ge 1/\sqrt{n}\right).\]
From Lemma \ref{lem:gaussian2}, we have that 
\[ \p \left( |s^Te_1|  \ge 1/\sqrt{n}\right) \ge \p(|Z| \ge 1)-1/n \]
where $Z \sim \mathcal{N}(0,1)$. Altogether, it follows that
\begin{align*}
    \p(\|A^{-1}\| \ge t)
    &\le \frac{\p_{A, s}\left( \|A^{-1}s\|_2 \ge t/\sqrt{n} \right)}{\p(|Z| \ge 1)-1/n} \\
    &\le\lemmafourpointforub{t/\sqrt{n}}
    \end{align*}
for $n>10$ as desired.
\end{proof}

\begin{lemma}\label{lem:gaussian}
Let $U, V \in \R^{n\times k}$ have independent $\mathcal{N}(0,1)$ coordinates. Let $M \in \R^{n \times n}$ be a matrix with singular values $s_1 \ge \cdots \ge s_n$ and let $w$ be a vector perpendicular to the first $n-1$ rows of $M + UV^T$. Then for $x_1 \ge 3\sqrt{\log(2nk)}, x_3 \ge nk$, and $x_2 \le 1$,  and all $t > 0$,
\[
\sup_{r\in\R}\p(|\inr{w}{Vu^n}-r|<1/t)\le
\lemmafourpointforub{t}
\]
for some $C,c>0$ where $u^n$ is the last row of $U$.
\end{lemma}
\begin{proof}
Let $m^1, \cdots, m^n$ be the rows of $M$ and let $u^1, \cdots, u^n$ be the rows of $U$. Then the rows of $A=M+UV^T$ are given by $m^i + Vu^i$. Let $y$ be the unit vector orthogonal to the span of $\{m^1, \cdots, m^{n-1}\}$. Consider the following three events:
\begin{align*}
    \mathcal{E}_1 &= \text{ event that every entry of  } U \text{ is at most }  x_1\text{ in absolute value}, \\
    \mathcal{E}_2 &= \text{ event that } \|V^Ty\| \text{ is at least } x_2, \\
    \mathcal{E}_3 &= \text{ event that } \|V\|^2 \text{ is at most } x_3,
\end{align*}
for $x_1 \ge 3\sqrt{\log (2nk)}, x_3 \ge nk,$ and $x_2 \le 1$.  Denote $\mathcal{E} = \mathcal{E}_1 \cup \mathcal{E}_2 \cup \mathcal{E}_3$.   We now show each of these occurs with high probability.
By a standard concentration bound, the maximum of $nk$ i.i.d.\ standard Gaussians is strongly concentrated around $\sqrt{2 \log(2nk)}$. In particular, $\mathcal{E}_1$ happens with probability at least $1-2\exp(-x_1^2/4)$.  Next, each coordinate of $V^Ty$ is distributed as $\mathcal{N}(0,1)$, and $\|V^Ty\|\ge\|V^Ty\|_\infty$, so we may upper bound $\p(\mathcal{E}_2^c)$ by $\p(|g|<x_2)^k$ where $g$ is $\mathcal{N}(0,1)$.  This means $\mathcal{E}_2$ occurs with probability at least $1-(2/\pi)^{k/2}x_2^k$. Lastly, the event $\mathcal{E}_3$ happens with probability at least $1-\exp(-\Omega(x_3))$ by Lemma \ref{lem:tail_s1}.

Now fix some realization of $V$ and $U$ such that $\event$ occurs.  Suppose for some parameter $z<\frac{\sqrt 3}2\frac{s_n}{x_1\cdot \sqrt{nk}}$ that we have $\|V^Tw\|<z$.  We will find a statement which this assumption implies, then take the contrapositive to obtain a lower bound on $\|V^Tw\|$.
From definition of $w$, we know
\[\inr{w}{m^i}+\inr{w}{Vu^i}=0\]
for $1 \le i \le n-1$.  We may apply Cauchy-Schwarz to $\inr{w}{Vu^i}=\inr{V^Tw}{u^i}$ and use event $\event_1$ to bound $\|u^i\|$ and obtain
 \[|\inr{w}{m^i}|\le\|V^Tw\|\cdot\|u^i\| \le z \cdot x_1 \cdot \sqrt{k}.\] 
Decompose $w = w^\| + w^\perp$ where $w^\|$ is in the span of $m^1, \cdots, m^{n-1}$ and $w^\perp$ is in the orthogonal complement. Write $w^\| = \sum_{i = 1}^{n-1} \alpha_i m^i$ for some coefficients $\alpha_i$. Then we have
\begin{align*}
\|w^\|\|^2 = |\langle w^\|, w \rangle|
  &\le\sum_{i=1}^{n-1} |\alpha_i| \cdot |\langle m^i, w \rangle |
\\&\le\|\alpha\|_1\cdot z\cdot x_1\cdot\sqrt{ k}
\\&\le\|\alpha\|  \cdot z\cdot x_1\cdot\sqrt{nk}
\\&\le\|w^\|\|\,\frac{z\cdot x_1\cdot\sqrt{nk}}{s_n}
\end{align*}
where $\alpha$ is the vector of $\alpha_i$'s, and the last step follows since $M$ is non-singular, making $\alpha$ the unique solution to $M^T\alpha=w^\|$. 
Now, note that $y$ and $w^\perp$ are parallel, so
\[
\|w^\perp\|\|V^Ty\|
  ={\|V^Tw^\perp\|}
\le{\|V^Tw\|+\|V^Tw^\|\|}
\le{z+\|V\|\|w^\|\|}
\le z\left(1 + \frac{x_1\cdot x_3^{1/2}\cdot\sqrt{nk}}{s_n}\right)
\]
where the last step follows from $\mathcal{E}_3$.
From $z<\frac{\sqrt 3}2\frac{s_n}{x_1\cdot \sqrt{nk}}$, we have
\[
\|w^\perp\|=\sqrt{1-\|w^\|\|^2}=\sqrt{1-\left(\frac{z\cdot x_1\cdot\sqrt{nk}}{s_n}\right)^2}>\frac12.
\]
Moving $\|w^{\perp}\|$ to the right hand side and using $\event_2$ and the above bound, we arrive at
\[
x_2\le\|V^Ty\|
< 2z\,\left(1 + {x_1\cdot x_3^{1/2}\cdot\sqrt{nk}}/{s_n}\right).
\]
We have thus established the syllogism
\[
\|V^Tw\|\le z\le\frac{\sqrt3}2\frac{s_n}{x_1\cdot \sqrt{nk}}
\implies
\frac12\frac{x_2s_n}{s_n + {x_1\cdot x_3^{1/2}\cdot\sqrt{nk}}}<z.
\]
By taking the contrapositive and setting
\(
z=\frac12\frac{x_2s_n}{s_n + {x_1\cdot x_3^{1/2}\cdot\sqrt{nk}}},
\)
we see that one of the inequalities on the left must fail.  It isn't the second one since $x_2<\sqrt3$, $x_3^{1/2}>1$, and $s_n>0$.  We therefore conclude our selection of $z$ lower bounds $\|V^Tw\|$.

This leads to sufficient anti-concentration.  For any fixed vector $x$, the distribution of $\inr{x}{u^n}$ is the same as $\mathcal{N}(0,\|x\|^2)$.  So, $\inr{V^Tw}{u^n}$ for random $V$ is a mixture of Gaussians, each of which have variance at least $z^2$ when we condition on $\event$.  The Levy concentration is thus easily bounded by
\[
\sup_{r\in\R}\p\left(|\inr{w}{Vu^n}+r|\le\frac1t\st\event\right)\le\frac{\sqrt{2/\pi}}{z\cdot t}.
\]
The desired bound then follows by incorporating a probability bound on the complement of $\mathcal{E}$.
\end{proof}

The following lemmas follow easily from basic properties of the SVD and Gaussian random variables so we omit their proof.

\begin{lemma}\label{lem:SVD}
	Consider a $n \times n$ matrix $M $and $u \in \mathcal{S}^{n-1}$ such that $\|M\| = \|Mu\|_2$. Then for every $v \in \mathbb{R}^n$, we have
	$$ \|Mv\|_2 \ge |u^Tv| \|M\|. $$
\end{lemma}

\begin{lemma}\label{lem:gaussian2}
	Let $x$ be a uniformly random vector in $\mathcal{S}^{n-1}$ and $Z \sim \mathcal{N}(0,1)$. Then for every $c > 0$,
	$$  \p \left( |x^Te_1|  \ge \sqrt{\frac{c}{n}}\right) \ge \p(|Z| \ge \sqrt{c})-\frac1n .$$
\end{lemma}

\section{Numerical experiments}\label{sec:numerical}

In this section, we numerically demonstrate our theoretical results by giving an example of a sparse family of $n$ by $n$ matrices that are `poorly' conditioned and whose condition number improves significantly after adding a random Gaussian rank $1$ perturbation. We show that this perturbation results in an improvement comparable to what is achieved after adding a dense Gaussian matrix while maintaining a low time complexity for matrix vector product operations.

Our family of $n$ by $n$ matrices will be constructed as follows: $M_n$ will have ones on the anti-diagonal and the first and third off-diagonals above the anti-diagonal. For example, $M_7$ is displayed below.

\[\begin{bmatrix}
\centering
0 & 0 & 0 & 1 & 0 & 1 & 1 \\
0 & 0 & 1 & 0 & 1 & 1 & 0 \\
0 & 1 & 0 & 1 & 1 & 0 & 0 \\
1 & 0 & 1 & 1 & 0 & 0 & 0 \\
0 & 1 & 1 & 0 & 0 & 0 & 0 \\
1 & 1 & 0 & 0 & 0 & 0 & 0 \\
1 & 0 & 0 & 0 & 0 & 0 & 0 
\end{bmatrix}
\]

It is shown in \cite{MO} that $M_n$ is ill-conditioned by showing that the magnitude of the smallest eigenvalue of $M_n$ is of the order $O(n/C^n)$ where $C \approx 1.47$ which implies that the smallest singular value of $M_n$ is also at most $O(n/C^n)$. The second smallest eigenvalue on the other hand is a constant.

\begin{figure}[!htbp]
    \centering
    \subfloat[]{{\includegraphics[width=7.75cm]{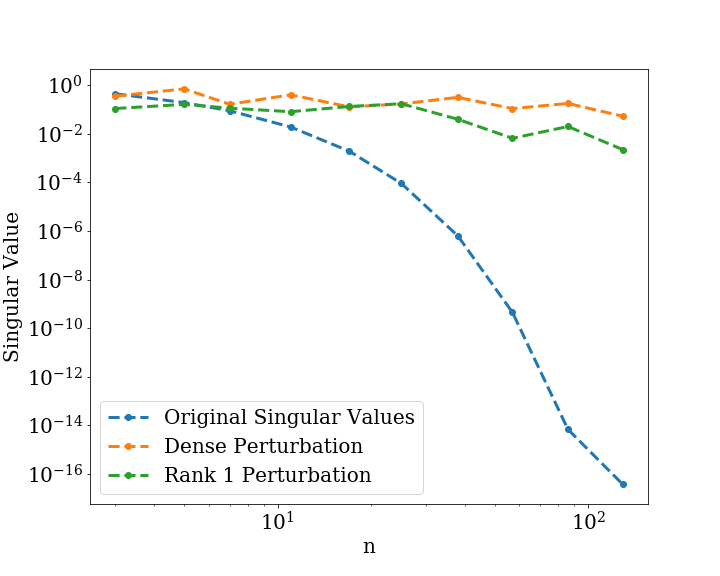} }}%
    \subfloat[]{{\includegraphics[width=7.75cm]{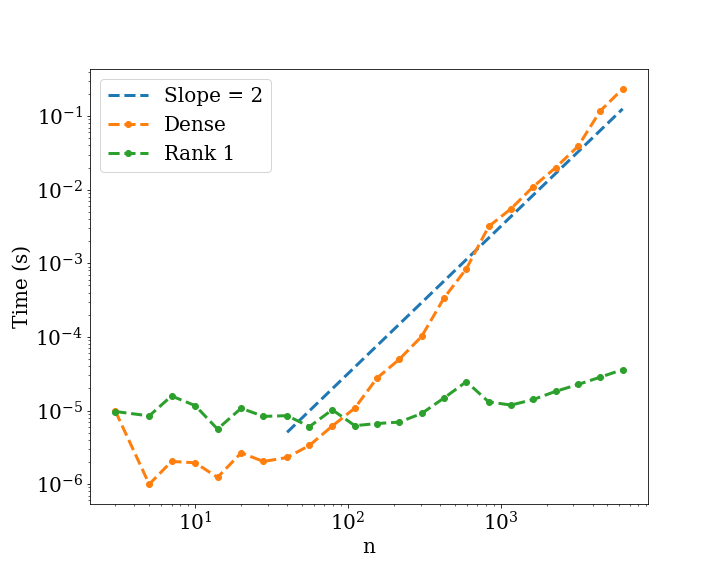} }}%
    \caption{(a) Smallest singular values of the original matrix compared against dense and rank $1$ perturbations. (b) Time taken to perform a matrix vector product after a dense perturbation and a rank $1$ perturbation. The cost for the dense perturbation has a quadratic scaling (slope $ = 2$).}%
    \label{fig:numerical}%
\end{figure}

In Figure \ref{fig:numerical} (a), we show the smallest singular value of $M_n$ for a range of $n$ along with the smallest singular values after a dense and rank $1$ perturbation. As we can see in the log-log plot, the original values are decaying exponentially while the smallest singular value after the rank $1$ perturbation is within a few orders of magnitude of the corresponding value after a dense perturbation. In Figure \ref{fig:numerical} (b), we show the time taken to perform a matrix vector product after a dense and a rank $1$ perturbation. For this task, we used the popular numerical libraries NumPy and SciPy. Since $M_n$ is sparse, it can be represented in a special sparse format to speed up computations. In the case of a rank $1$ perturbation, we only need to store two additional vectors and a matrix vector product $(M+uv^T)x$ can be performed as 
\[(M_n+uv^T)x = M_nx + (v^Tx) \cdot u. \]
However, in the case of a dense perturbation, we need to store a dense matrix $G$ and perform the matrix vector product operation with a vector and a dense matrix resulting in a much slower operation than in the rank $1$ case. Indeed, note that the slope of the `Dense' curve in Figure \ref{fig:numerical} (b) is close to $2$ signifying a quadratic increase in time. Overall, we see that in this case, a rank $1$ update results in a comparable improvement of the condition number of $M_n$ while greatly improving the cost to perform a fundamental matrix operation.

\section{Low-rank noise model for other problems in smoothed analysis } \label{sec:challenges}
In this section, we outline some of the challenges that arise when applying the rank $1$ noise model in other popular problems studied in smoothed analysis. While not a comprehensive survey of all problems, we focus on two of the most studied applications of this framework outside of the condition number. These are the simplex method and $k$-means. For these problems, the standard noise model is the dense one where every entry of the input matrix or input set of points respectively, is independently perturbed by a random Gaussian. We highlight some of the challenges that arise when trying to carry out existing proof techniques for these problems using rank $1$ noise. This ultimately shows that new ideas are required to bypass the lack of independence as we did for the condition number.

\subsection{Simplex method}
The simplex method is one of the most famous applications of the smoothed analysis framework. The goal is to solve a linear program of the form $\max c^Tx$ subject to $Ax \le b$ using the simplex method where the entries of $A \in \mathbb{R}^{m \times n}$ have been perturbed by random noise. Recall that the simplex method operates by moving among the vertices of the polytopes defined by the constrained matrix $A$. The geometric operation of moving from one vertex to another is called a pivot operation and the most commonly analyzed pivot operation with respect to smoothed analysis is the shadow vertex pivot method. 

Without getting into technical details that will lead us too far afield, we note that the shadow vertex pivoting method requires us to calculate the following bound: let $a_i$ for $1 \le i \le m$ denote the rows of the matrix $A$ and let $W$ be a fixed two dimensional subspace. We wish to bound
\[\E[|\text{edges}(\text{conv}(a_1, \cdots, a_m) \cap W)|] \]
where $\text{conv}(a_1, \cdots, a_m)$ is the convex hull of the rows (see \cite{friendlysmooth} for more information). 

To calculate the above bound, we essentially need to understand the probability that $a_j^T \theta \le t$ for a range of values of $j$ and some $t \in \mathbb{R}$ (here $\theta$ represents the normal vector of the line connecting some two points $a_i, a_k$. For the pair $a_i, a_k$ to be on the convex hull, we need the rest of the points to be on one side of the line). In the case that we add independent noise across the rows, this bound is possible to compute due to independence across $a_j$. However, in the case that we add rank $1$ noise $u^Tv$ (here $u \in \mathbb{R}^m, v \in \mathbb{R}^n$) to $A$, these probabilities become intractable using existing methods since $a_j$ satisfying  $a_j^T \theta \le t$ gives us information about all other $a_j'$ for $j' \ne j$ since randomness is shared across the rows.

Nevertheless, it is possible to get a weak result for the smoothed analysis of the simplex method in our low-rank noise model by using a different pivoting operation. It is shown in \cite{geometric_edge} that if the rows satisfy a certain \emph{geometric} property, then using a random pivoting rule results in an expected polynomial number of steps for the simplex method to converge. 

The geometric property is the following:
For any $I \subseteq [m]$, and $j \in [m]$, if $a_j$ is not in the span generated by $a_i, i \in I$, then the distance from $a_j$ to this span is at least $\delta$. We note that the bound on the expected number of steps depends polynomially on $1/\delta$ and other parameters. This geometric property reduces to a singular value estimate as follows. For simplicity, lets focus on $j=1$ and $I = \{2, \cdots, n-1\}$. As in Section \ref{sec:previous}, it follows that $\|A_{[n]}^{-1}e_1\|$ is equal to $1/|w^Ta_1|$ where $w$ is the normal vector of the span of the rows $a_2, \cdots, a_n$. Thus, if $s_n(A_{[n]})$ is `not too small' then $\|A_{[n]}^{-1}e_1\|$ cannot be `too large' and consequently, the distance from $a_1$ to the span of $a_2, \cdots, a_n$ is `not too small' (we are intentionally leaving our specific relations for a high level overview). The caveat is that we need the geometric property to hold between $a_1$ and \emph{every} set of $n-1$ other vectors. However, since the bound of Theorem \ref{thm:main_rank_k} only gives us an inverse polynomial probability, we cannot afford the union bound of $\binom{m}n$ unless $m = n+C$ for some constant $C$, which is not a realistic scenario. 

Lastly, empirical evidence shows that rank $1$ perturbation may not be a suitable if the \emph{original} simplex method (the Dantzig simplex method) is used. In Figure \ref{fig:simplex}, we use the Klee-Minty lower bound \cite{KM} for the Dantzig simplex method and add either a dense Gaussian or a rank $1$ perturbation to the constraint matrix. We then plot the average number of pivot steps taken over twenty independent trials. It can be seen that a rank $1$ perturbation only slightly improves over the exponential number of pivot steps required by the Dantzig simplex method whereas dense perturbations help greatly. 

\begin{figure}[!htbp]
\includegraphics[width=10cm]{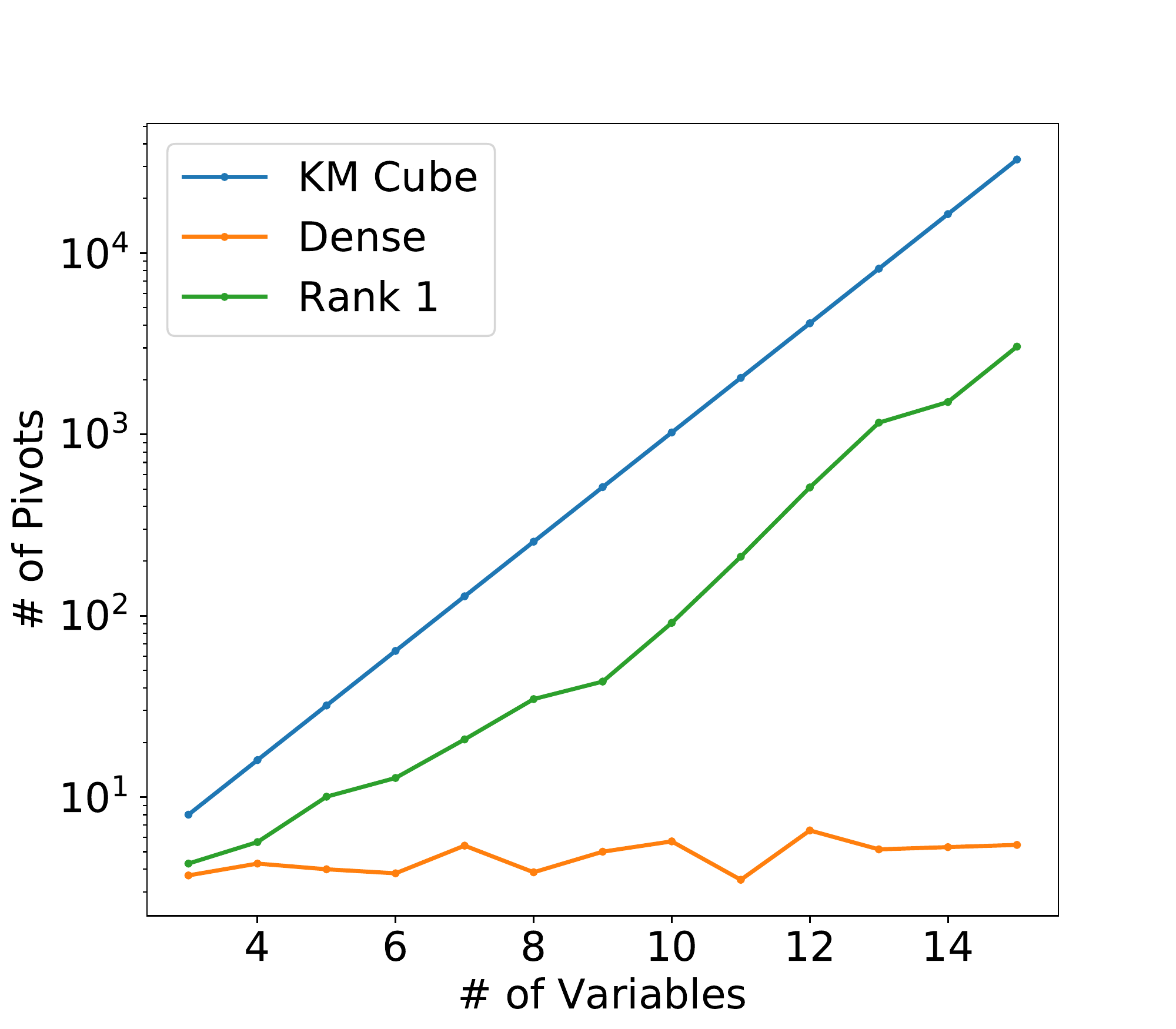}
\centering
\caption{The Dantzig simplex method applied to the Klee-Minty lower bound and its random perturbations.}
\label{fig:simplex}
\end{figure}

We conclude our discussion with a major open problem.
\begin{open}
Is there a pivoting rule for the simplex method that runs in expected polynomial time if we add random rank $1$ noise to the constraint matrix?
\end{open}

\subsection{$k$-means clustering}
Recall that in the $k$-means problem, we are given a set $X$ of $n$ points in $\mathbb{R}^d$ and our goal is to partition the points into $k$ sets $S_i$ to minimize the objective
\[\sum_{i=1}^k\sum_{x \in S_i} \|x-\mu_i\|^2 \]
where $\mu_i$ is the mean of the points in $S_i$. A common heuristic for this problem, also confusingly known as the $k$-means algorithm, or Lloyd's method, is to randomly pick an initial set of $k$ centers, assign each point in $X$ to its closest center, update the means accordingly, and repeat until convergence. In the smoothed analysis framework, it was shown that if each point in $X$ is perturbed by an independent Gaussian vector then convergence happens in polynomially many steps \cite{smoothedkmeans_best}. The existing analysis all crucially rely on the following geometric lemma.

\begin{lemma}\label{lem:kmeans_prior}
Let $x \in \mathbb{R}^d$ be drawn according to a $d$-dimensional Gaussian distribution of standard deviation $\sigma$, and let $B$ be the $d$-dimensional ball of radius $\eps$ centered at the origin. Then $\p(x \in B)\le (\eps/\sigma)^d$.
\end{lemma}
This lemma roughly states that the probability of a random Gaussian being in any ball of radius $\varepsilon$ is at most $\varepsilon^d$, and is used to union bound over exponentially many events in the smoothed analysis of $k$-means.

Surprisingly, this lemma \emph{does not} hold in our `rank $1$' setting. More precisely, we can prove the following probabilistic bound which is a major impediment to understanding the smoothed complexity of the $k$-means problem with rank $1$ noise.

\begin{lemma}\label{lem:kmeans}
Let $x \in \mathbb{R}^d$ be drawn according to a standard $d$-dimensional Gaussian distribution and let $y \in \mathbb{R}$ be a scalar standard Gaussian random variable. If $B$ is the $d$-dimensional ball of radius $\eps$ centered at the origin then $\p(yx \in B) = O(\eps/\sqrt{d})$.
\end{lemma}
Note that $yx \in \mathbb{R}^d$. We are considering random variables of this form because if a rank $1$ perturbation was added to $X$, then each row is perturbed by a random vector of the form $yx \in \mathbb{R}^d$. Lemma \ref{lem:kmeans} roughly states that the probability that the random vector $yx$ is in any ball of radius $\eps$ only weakly depends on the dimension $d$. In particular, we do not get an exponentially small probability afforded by Lemma \ref{lem:kmeans_prior} that enables us to union bound over exponentially many events as in the arguments for the smoothed analysis of $k$-means under the standard noise model. 

The intuition for Lemma \ref{lem:kmeans} is as follows. First, note that from standard Gaussian concentration, we have $\|x\| \approx \sqrt{d}$. Treating this as fixed for now, this means that $y\|x\|$ is approximately distributed as a scalar Gaussian distribution with variance $d$. Therefore, from Proposition \ref{prop:gaussian_ball}, it follows that $\p(|y|\|x\| \le \eps) = \Theta(\eps/\sqrt{d})$. We now formalize this argument.

\begin{proof}
Note that $\|x\|_2^2$ is a chi-squared variable with $d$ degrees of freedom. From \cite{product_dist}, we know that the density of the product $Z = \|yx\|_2^2 = y^2 \|x\|_2^2$ is given by
\[f_Z(z) \simeq \frac{1}{2^{d/2}\Gamma(d/2)}\int_0^{\infty}\left(x^{d/2-2} e^{-x/2}\right) \left( \frac{1}{\sqrt{z/x}}e^{-z/(2x)} \right) dx.   \]
Therefore, 
\[
    \p(\|yx\|_2^2 \le \eps^2) \simeq \frac{1}{2^{d/2}\Gamma(d/2)}\int_0^{\eps^2} \int_0^{\infty}\left(x^{d/2-2} e^{-x/2}\right) \left( \frac{1}{\sqrt{z/x}}e^{-z/(2x)} \right) dx \ dz. \]
We now switch the order of summation which is valid since the integrand is positive. From the definition of the error function, we can check that 
\[ \int_0^{\eps^2}\frac{1}{\sqrt{z/x}}e^{-z/(2x)} \ dz \simeq x \cdot \text{erf}(\eps/\sqrt{x}).  \]
We now use the estimate $\text{erf}(t) \le 2t$ which holds for all $t \ge 0$. This gives us
\[
    \int_0^{\infty} x^{d/2-1} e^{-x/2} \text{erf}(\eps/\sqrt{x}) \ dx \lesssim \eps \int_0^{\infty} x^{d/2-3/2} e^{-x/2} \ dx
    =\eps 2^{d/2-1/2}\Gamma(d/2-1/2).
\]
Finally, noting that $\Gamma(d/2-1/2)/\Gamma(d/2) \lesssim 1/\sqrt{d}$ gives us our desired probability bound.
\end{proof}
Note that the above bound is the best that we can hope for. Indeed, we can say that $\|x\|_2^2 = \Omega(d)$ with probability $1/2$ so conditioning on this event, we have that $\Pr(|y|\|x\|_2 \le \eps) = \Omega(\eps/\sqrt{d})$. We note that Lemma \ref{lem:kmeans_prior} is also required for the smoothed analysis of other Euclidean problems such as a local search heuristic for Euclidean TSP \cite{TSP}.

\paragraph{Acknowledgements.}
We thank Sushruth Reddy and Samson Zhou for helpful conversations. We also thank Piotr Indyk and Arsen Vasilyan for helpful feedback on a draft of the paper.

\bibliographystyle{abbrv}
\bibliography{mainbib}

\appendix

\end{document}